%% file: draft.tex
\documentclass[letterpaper,11pt]{article}
\usepackage{fullpage}
\usepackage{graphicx}
\usepackage{comment}
\usepackage{mathtools}
\usepackage[linesnumbered,ruled,vlined]{algorithm2e}
\usepackage[utf8]{inputenc}
\usepackage{tikz}
\usepackage{amsthm}
\usepackage{subcaption}
\usepackage{amsfonts}
\usepackage{hyperref}
\usepackage[framemethod=TikZ]{mdframed}
\usepackage{amsthm}
\usepackage{thmtools}
\usepackage{thm-restate}
\usepackage[normalem]{ulem}
\usepackage{float}
\usetikzlibrary{arrows,shapes,snakes,automata,backgrounds,petri,calc}

\newtheorem{theorem}{Theorem}
\newtheorem{lemma}[theorem]{Lemma}

\newtheorem{observation}{Observation}

\usepackage{natbib}
\newtheorem{example}{Example}[section]

\input{sources/macros}
\definecolor{bostonuniversityred}{rgb}{0.8, 0.0, 0.0}

\newcounter{theo}[section] 

\begin{document}
	
	\title{Computing Stackelberg Equilibria of Large General-Sum Games\footnote{This work was supported in part by the National Science Foundation under grants CCF-1733556 and CCF-181501, NSF CAREER award CCF-1053605,  NSF BIGDATA grant IIS-1546108, NSF AF:Medium grant CCF-1161365, DARPA GRAPHS/AFOSR grant FA9550-12-1-0423, and another DARPA SIMPLEX grant.}}
\author{
	Avrim Blum\thanks{Toyota Technological Institute, Chicago, IL 60637, avrim@ttic.edu}
	\and Nika Haghtalab\thanks{Microsoft Research, Cambridge, MA 02142, nika.haghtalab@microsoft.com}
	\and MohammadTaghi Hajiaghayi\thanks{University of Maryland, College Park, MD 20742, \{hajiagha,saeedrez\}@cs.umd.edu}
	\and Saeed Seddighin\footnotemark[4]
}
	\date{}
\maketitle
\begin{abstract}
	\input{sources/abstract}

\end{abstract}

\input{sources/intro}

\input{sources/related}

\input{sources/preliminaries}

\input{sources/alternative}

\input{sources/apxhardness_new}
\bibliographystyle{ACM-Reference-Format}	
\bibliography{sources/ne,sources/Nika_bib} 
\clearpage
\appendix
\input{sources/app_example}

\input{sources/apx-appendix}
\end{document}

%% file: sources/macros.tex
\newcommand{\I}{\mathsf{I}} 
\newcommand{\cori}{\mathsf{Cor(I)}} 
 
\newcommand{\common}{\mathsf{common}} 
\newcommand{\dist}{\mathsf{dist}}

\newcommand{\strategiesa}{\mathsf{S_L}}
\newcommand{\strategiesb}{\mathsf{S_F}}
\newcommand{\strategiesl}{\mathsf{S_L}}
\newcommand{\strategiesf}{\mathsf{S_F}}
\newcommand{\payoffa}{\mathsf{U_L}}
\newcommand{\payoffb}{\mathsf{U_F}}
\newcommand{\payoffl}{\mathsf{U_L}}
\newcommand{\payofff}{\mathsf{U_F}}
\newcommand{\Le}{\mathsf{L}}
\newcommand{\Fo}{\mathsf{F}}

\newcommand{\opti}{\mathsf{Opt}(\I)}
\newcommand{\vali}{\mathsf{Val}(\I)}

\newcommand{\poly}{\mathrm{poly}}

\newcommand{\mgame}{\textsc{Permuted Matching}}
\newcommand{\tgame}{\textsc{Incentive Games}}
\newcommand{\matchingproblem}{$\pi$\textsc{-transformation-identical-matching}}
\newcommand{\petrankhardness}{Petrank}
\newcommand{\threedM}{\textsc{Maximum 3D Matching}}

\DeclareMathOperator{\E}{\mathbb{E}}

%% file: sources/abstract.tex
We study the computational complexity of finding Stackelberg Equilibria in general-sum games, where the set of pure strategies of the leader and the followers are exponentially large in a natrual representation of the problem.


In \emph{zero-sum} games, the notion of a Stackelberg equilibrium coincides with the notion of a \emph{Nash Equilibrium}~\cite{korzhyk2011stackelberg}.
Finding these equilibrium concepts in zero-sum games can be efficiently done  when the players have polynomially many pure strategies or when (in additional to some structural properties) a best-response oracle is available~\cite{ahmadinejad2016duels, DHL+17, KV05}. Despite such advancements  in the case of zero-sum games, little is known for general-sum games.

In light of the above, we examine the computational complexity of computing a Stackelberg equilibrium in large general-sum games. 
We show that while there are natural large general-sum games where the Stackelberg Equilibria can be computed efficiently if the Nash equilibrium  in its zero-sum form could be computed efficiently, in general, structural properties that allow for efficient computation of Nash equilibrium in zero-sum games are not sufficient for computing Stackelberg equilibria in general-sum games.

%% file: sources/intro.tex
\section{Introduction}
\allowdisplaybreaks

Recent years have witnessed significant interest in Stackelberg games and their equilibria.
A \emph{Stackelberg  game} models an interaction between two players, a \emph{leader} and a \emph{follower}, where the leader's goal is to \emph{commit} to a randomized strategy that yields the highest utility, given that the follower responds by choosing an action that is best for itself. Such a pair of strategies is called a \emph{Stackelberg equilibrium (SE)}.
The interest in these games is driven, in part, by their applications to security~\citep{tambe2011security} and their adoption by major security agencies such as the US Coast Guard, the Federal Air Marshals Service, and the Los Angeles Airport Police.

Standard approaches for finding a Stackelberg equilibrium, such as the Multiple LPs approach of \citet{conitzer2006computing}, run in time polynomial in the number of pure strategies of the leader and follower.
As Stackelberg games and their applications have become more prevalent, they are increasingly used to model complex scenarios where one or both players'  strategy sets are \emph{exponentially large in a natural representation of the problem}, in which case existing approaches are not computationally feasible.
In this work, we consider such ``large'' games and  ask whether there are  computationally efficient algorithms for finding their Stackelberg equilibria.

Of course, such algorithms cannot exist without some assumptions on the problem structure. Here, we review the common assumptions and approaches for computing minimax-optimal solutions in large {\em zero-sum} games, where minimax strategies, Nash equilibria, and Stackelberg equilibria all coincide.
Computing these equilibrium concepts in $2$-player zero-sum games has received significant attention~\citep{immorlica2011dueling,ahmadinejad2016duels,behnezhad2016faster,garg2011bilinear,xu2016mysteries,wang2017security,DHL+17,von1945theory,hannan1957approximation, freund1995desicion,littlestone1994weighted,conitzer2006computing}.
For large zero-sum games, two structural assumptions that have proven useful in computing a Nash equilibrium are the ability to efficiently optimize a linear function over the strategy space of each player~\citep{ahmadinejad2016duels} and the ability to compute the best-response of each player against a mixed strategy of the other combined with a decomposibility property of the action set~\citep{DHL+17}.

In general-sum games, however, Stackelberg, Nash, and Minimiax equilibria diverge (See Appendix~\ref{app:example} for an illustration).
In general-sum games the leader can benefit from committing to a mixed strategy and obtain a more favorable Stackelberg equilibrium than any Nash equilibrium. 
From the algorithmic perspective, a Stackelberg equilibrium in a general-sum game can be computed efficiently when the game is small.
That is, there are algorithms, such as the Multiple LPs approach of~\citet{conitzer2006computing}, that run in time $\poly(|\strategiesl|, |\strategiesf|)$ where  $\strategiesl$ and $\strategiesf$ are the set of pure strategies of the leader and follower, respectively.
While this method is an efficient approach for computing a Stackelberg equilibrium in small games, it become computationally inefficient in many natural scenarios where the set of actions of the leader or follower is exponential in a natural representation of the game. Examples of such settings include games inspired by applications to security, where either the actions of the leader or the follower represent sets of edges in a graph. As opposed to the zero-sum case for which existence of certain structural properties are known to lead to efficient computation of the equilibrium  concepts, computation of Stackelberg equilibrium in large general-sum games has remained mostly unexplored.

\subsection{Our Results and Contributions} 

In light of the above, we examine the computational complexity of computing Stackelberg equilibria in large general-sum games. Specifically, we consider two classes of general-sum games, both of which demonstrate structural properties that under the zero-sum assumption would lead to efficient algorithms for computing the minimax optimal strategies. 
For the first class of games, we give an efficient  algorithm for computing a Stackelberg Equilibrium. In the second class of games, we show that even approximating the Stackelberg equilibrium is NP-Hard.  This drives home the main message of this work, that is
\emph{while there are natural large general-sum games where the Stackelberg Equilibria can be computed efficiently if the Nash equilibrium  in its zero-sum form could be computed efficiently, in general, structural properties that allow for efficient computation of Nash equilibrium in zero-sum games are not sufficient for computing Stackelberg equilibria in general-sum games.}

In more details, the two classes of games we work with are as follow.

\paragraph{\tgame}
In Section~\ref{sec:tax}, we introduce a class of games called \tgame. In these games, the actions of the leader can be described as two-part actions, the first part of the action is an element of a set and the second part of the action is a set of incentives to the follower for playing certain actions.

As a motivating example consider a taxation scenario. In this setting, a government agency (e.g., IRS) takes the role of the leader and a taxpayer is the follower.  A number of investments, indicated by the set $E$, are available to the taxpayer. Each investment $e$ has a return of $c_e$ to the taxpayer. The taxpayer invests in a package of investments $S\subseteq E$ that has the highest net payoff. 
The government agency is interested in taxing these investments in order to maximize the tax revenue. To do so, the agency allocates $1$ unit of taxes\footnote{More generally instead of $1$ unit of tax one can consider any other fixed upperbound for the total amount of taxes used by the system. Note that any modern tax system is design with such an upperbound in mind to ensure the welfare of society and avoid  financial or social crises caused by excessive taxes.} between these investments. 
There are two types of taxation mechanisms. First is taxing an individual investment $e$ by some amount $x_e$. The second is to provide tax relief $v_S$ for a package of options the taxpayer has invested in. Examples of the second type of taxation mechanism include United States federal residential renewable  energy tax credit that offers a tax break to individuals who have invested in home electric power storage, e.g., batteries, and home-generated renewable energy, e.g., solar panels, but no tax break to those who have invested in the former without the latter~\citep{solar, solarexplain}. The tax revenue and the net payoff the taxpayer respectively receive from individual taxes $x_e$ and combinatorial tax reliefs  $v_S$ when the taxpayer invests in investments $S$ are 
$\sum_{e\in S} x_e - v_S$ and 
$ \sum_{e\in S} (-x_e + c_e) + v_S.$
It is not hard to see that these tax breaks play an essential role in the design of tax systems. Not only they increase the total tax revenue obtainable by a tax system (See Example~\ref{eg:commit} for an illustration) but  they can also be used to incentivize the taxpayers to take actions that are more beneficial to the government.

More generally, we consider Stackelberg games and  we consider a family of sets $\S\subseteq \{0,1\}^{E}$ and leader and follower element payoffs, $C_e$ and $c_e$, respectively, for all $e \in E$.
A pure strategy of the leader is to choose $e\in E$  and a vector of \emph{incentives} $\vec v \in [0,1]^{|\S|}$, such that $\|\vec v\|_0 \in \mathrm{poly}(|E|)$. \footnote{The sparsity requirement is such that the leader can communicate its strategy to the follower efficiently.} A pure strategy of the follower is to choose one set $S\in \S$. 
The payoff of the leader and follower are defined, respectively, by 
\begin{align*}
\payoffl( (e, \vec v),  S) &= 1_{e\in S} - v_S + C_e,\\
\payofff((e, \vec v), S)   &=  - 1_{e\in S} + v_S + \sum_{e'\in S} c_{e'},
\end{align*} 
that is, the players receive non-zero-sum utilities from their individual choices, i.e., $C_e$ and $\sum_{e\in S} c_{e'}$, and zero-sum utilities from choosing actions that intersect, i.e., $\pm 1_{e\in S}$, and from the incentives provided on the followers actions sets, i.e., $\pm v_S$.

We first note that when $c_e$ and $C_e$ are set to $0$ for all $e\in E$, this game is zero-sum and can be efficiently solved when each player can compute its best-response to any choice of mixed strategy of the other player, i.e., optimize a linear function over the strategy space of the other player using existing results~\citep{ahmadinejad2016duels, DHL+17, KV05}. 

When $c_e$ and $C_e$ are non-zero, we show that the leader can obtain a higher payoff equilibrium if it could make additional commitments in the form of incentives for the follower, i.e., can play non-zero $\vec v$.
%
An interesting aspect of this game is that it is derived from a simple Stackelberg game model (where $\vec v = \vec 0$) by adding zero-sum payoffs that \emph{only benefit the follower}. Yet, the leader's payoff in the Stackelberg equilibrium of the new game is much higher than its payoff in the original game.
 Moreover, as we show in Theorem~\ref{thm:main-discount} there is a polynomial time algorithm for finding the Stackelberg equilibrium of such games  when the leader can optimize a linear function over the actions of the follower, which is a similar condition to the ones used for computing Stackelberg equilibria in large zero-sum games~\citep{ahmadinejad2016duels, DHL+17, KV05}.
%
%

\paragraph{\mgame\ \textsc{Game}}
In Section~\ref{section:apx}, 
we introduce a non-zero-sum game called \mgame. In this game, there is a graph $G = (V, E)$ and a permutation $\pi:E \rightarrow E$. The set of pure strategies of the leader and follower is the set of all matchings in $G$. The goal of the leader is to maximize the intersection of its matching with the $\pi$-transformation of the matching of the follower. On the other hand, the goal of the follower is to maximize the intersection of the two matchings, with no regards to $\pi$. More formally, for $S\subseteq E$ we define $\pi(S) = \{e \in S| \pi(e)\}$. Then the utility of the leader and follower are defined, respectively, by 
\begin{equation*}
\payoffl(M_1,M_2) = |M_1 \cap \pi(M_2)| \qquad \qquad \payofff(M_1,M_2) = |M_1 \cap M_2|.
\end{equation*} 


It is not hard to see that, in this game,
the problem of finding a best response for a player reduces to computing maximum weighted matching of $G$ and can be solved in polynomial time. This would have been sufficient for getting a polynomial time algorithm for finding a Stackelberg equilibrium had the game been a zero-sum~\citep{ahmadinejad2016duels, DHL+17, KV05}. In a sharp contrast, however, we show that computing a Stackelberg equilibrium of this general-sum game is APX-hard, even though, we can compute player's best-response efficiently.

We obtain this hardness result via two reductions. First, we define the following computational problem:
\begin{quote}
\matchingproblem: Given a graph $G$ and a mapping $\pi: E(G) \rightarrow E(G)$ over the edges of $G$,  find a matching $M$ of $G$ that maximizes $|M \cap \pi(M)|$. 
\end{quote}
We next show that  computing  an approximate Stackelberg equilibrium of the \mgame\ game is at least as hard as computing an approximate solution for the \matchingproblem\enspace problem. The crux of the argument is that if in an instance of the \matchingproblem\enspace problem there exists a matching which is almost identical to its $\pi$-transformation, then a Stackelberg equilibrium of the \mgame\ game is closely related to that matching. Thus, any solution for the \mgame\ game can be turned into a solution for \matchingproblem\enspace with almost the same quality. In the second step, we reduce the \matchingproblem\enspace problem to the \threedM\ problem, which  is known to be APX-hard~\citep{petrank1994hardness}.

We note that our results strengthen the existing hardness results of \citet{letchford2010computing,licatcher} that showed that computing Stackelberg equilibrium is \emph{NP-hard}\footnote{Interestingly, it is not hard to show that player best-response can also be computed efficiently in the games used by \citet{letchford2010computing,licatcher}, although this was not central to their results.}. Our APX-hardness result shows that one cannot even approximate the Stackelberg equilibria of large games within \emph{an arbitrary constant factor}, even when best-response can be efficiently computed.


%% file: sources/related.tex
\subsection{Related Work}
There is an extensive body of work investigating the complexity of solving Security games, which is a special case of computing Stackelberg equilibria (see \textit{e.g.} ~\cite{tambe2011security,behnezhad2017polynomial,xu2014solving,basilico2009leader,letchford2011computing,xu2016mysteries}).

\paragraph{zero-sum games}
Several algorithms have been proposed for finding the Stackelberg equilibria of a special case of security games called the spatio-temporal security games \citep{behnezhad2017polynomial,xu2014solving}. These games are zero-sum by definition, where Stackelberg equilibria, Nash equilibria, and Minimax equilibria all coincide. In comparison, our work focuses on general-sum games.

\paragraph{Smaller general-sum games}
Several works have introduced polynomial time algorithms for computing  Stackelberg equilibria in games where \emph{only one player's} strategy set is exponentially large ~\citep{kiekintveld2009computing,xu2016mysteries}. A common approach used in this case is the Multiple LPs approach of~\citet{conitzer2006computing} that runs in $\poly(|\strategiesl|, |\strategiesf|)$. In this approach one creates a separate Linear Program for every action $y\in \strategiesf$ of the follower, where the variables represent the probability assigned to the actions of the leader, the objective maximizes the expected payoff of the leader, and the constraints assure that action $y$ is the best-response of the follower. 
This method can be implemented efficiently even when the leader's strategy set is exponentially large, e.g., when a separation oracle can be implemented efficiently. In comparison, our main computational result focus on settings where both the leader and follower have exponentially many strategies.

\paragraph{Existing hardness results}
\cite{letchford2010computing} studied the computational complexity of extensive form games and proved a closely related hardness result.

They showed that computing Stackelberg equilibrium of a game is weakly NP-hard using a reduction from Knapsack. Interestingly, one can efficiently compute player best-response in their setting. In comparison, our hardness result improves over these results by showing that Stackelberg equilibria are hard to approximate within arbitrary constant factor even when player best-response can be computed efficiently.

A number of works have investigated the relationship between the Stackelberg equilibria and Nash equilibria of security games and have shown that computing a Stackelberg equilibrium is at least as hard as computing a Nash equilibrium of general-sum games.
\cite{korzhyk2011security} studied a special class of general-sum Stackelberg Security games where any Stackelberg Equilibrium is also a Nash equilibrium. This shows that computing Stackelberg equilibria is harder than computing Nash equilibria.
\cite{licatcher} studied Bayesian Stackelberg Games, where there is additional uncertainty about the attacker and show that computing the Stackelberg equilibrium is hard, and introduce an exponential time algorithm for computing the Nash equilibria.
In comparison, our work shows that Stackelberg equilibria are hard to approximate even when players best-response is easy to compute. That is, we show a gap between the computational complexity of approximating Stackelberg equilibrium of a general-sum game and that of its corresponding zero-sum variant.

%% file: sources/preliminaries.tex
\section{Preliminaries}\label{sec:preliminaries}


Throughout this paper, we study Stackleberg equilibria of large games. Our emphesis is on two player games and therefore we denote the players by $\Le$ (leader) and $\Fo$ (follower). Let $\strategiesa$ and $\strategiesb$ be the set of actions (pure strategies) of players $\Le$ and $\Fo$. For a pair of pure strategies $x \in \strategiesa$ and $y \in \strategiesb$, we denote the payoffs of players $\Le$ and $\Fo$ by $\payoffa(x,y)$ and $\payoffb(x,y)$, respectively. Similarly, for a pair of mixed strategies $X$ and $Y$ we denote the payoffs by 
\begin{align*}
\payoffa(X,Y) &= \E_{x \sim X, y \sim Y} [\payoffa(x,y)]\\
\payoffb(X,Y) &= \E_{x \sim X, y \sim Y} [\payoffb(x,y)].
\end{align*}

In Stackelberg games, the leader commits to a (possibly mixed) strategy $X$ and plays this strategy. The follower then plays a best response against $X$, $b(X)$, according to her payoff function. Since the follower goes second its best-response is a deterministic action $b(X) = \max_{y} \payofff(X, y)$.
In case there is more than one best response for the follower, we assume she plays the one that maximizes the payoff of the leader. A pair of strategies $X$ and $y$ are in Stackelberg equilibrium if $y$ is a best response of the follower against $X$ and $X$ maximizes the payoff of the leader, subject to the follower playing a best response.

%% file: sources/alternative.tex
\section{\tgame}  \label{sec:tax}

In this section, we discuss a class of Stackelberg games where the
leader has the ability to make additional commitments in the form
of additional incentives to the follower. Recall that a natural scenario that can be addressed by this Stackelberg model is \emph{taxation}.
In this case the leader can set taxes on individual investments but can also provide tax breaks on bundles of investments that the  tax payer has invested in.
We first show how these additional combinatorial  incentives can improve the leader's payoff significantly and then show polynomial time algorithms for computing a Stackelberg equilibrium in this model.

Let us first recall the definition of \tgame. In this model, we consider a set of elements $E$, a family of its subsets $\S\subseteq \{0,1\}^{E}$ and rewards $C_e$ and $c_e$ for all $e \in E$.
The set of pure strategies of the leader is  $\strategiesl = E \times [0,1]^{|\S|}$. That is, each action of the leader has two parts, the first part is an element  $e\in E$ and the second part is a  vector of \emph{incentives} $\vec v \in [0,1]^{|E|}$. We assume that the leader is restricted to playing incentive vectors $\|\vec v\|_0 \in \mathrm{poly}(|E|)$. \footnote{The sparsity requirement is such that the leader can communicate its strategy to the follower efficiently.}
The follower's strategy set is $\strategiesf = \S$. The leader and follower payoffs are as follows.
\begin{align}
\payoffl( (e, \vec v),  S) = 1_{e\in S} - v_S + C_e,\text{ and } \payofff((e, \vec v), S) =  - 1_{e\in S} + v_S + \sum_{e'\in S} c_{e'},
\end{align} 

that is, the players receive non-zero-sum utilities from their individual choices, i.e., $C_e$ and $\sum_{e\in S} c_{e'}$, and zero-sum utilities from choosing actions that intersect, i.e., $\pm 1_{e\in S}$, and from the incentives provided on the followers actions sets, i.e., $\pm v_S$. For ease of exposition and by the linearity of the payoffs, we denote a mixed strategy of the leader by $(\vec x, \vec V)$, where $x_e$ is the probability with which the first part of the leader's action is $e$ and $V_S$ is the expected incentive provided on action $S$ in the second part of the leader's action. Note that in this case, the expected utilities of the leader and follower are
\begin{align}
 \payoffl( (\vec x, \vec V), S) &=  \sum_{e\in S} x_e - V_S+ \sum_{e\in E} x_e C_e  \\
  \payofff((\vec x, \vec V), S) &=  \sum_{e\in S} (-x_e + c_e) + V_S.
\end{align}

Let us first consider a variation of \tgame\ where the leader cannot  provide additional incentives to the follower, i.e, $\strategiesl = E \times \vec 0$. The only difference between these games is that \tgame\ are amended by allowing zero-sum non-negative payments $\vec v$ that benefit the follower solely. One might wonder if the commitment to make additional payments $\vec v$ to the follower can ever be beneficial to the leader.
%
%
This is exactly what we demonstrate in the next example. That is, by allowing the leader to make additional zero-sum payoffs that only benefit the follower, we can obtain Stackelberg equilibria that have much higher payoff to the leader.

\begin{example} \label{eg:commit}
Consider a graph instance in Figure~\ref{fig:example-commit}, $E$ is the set of all edges, $\S$ is the set of all $s$-$t$ paths, there are no edge payoff to the leader, i.e., $C_e = 0$ for all $e\in E$, and the edge payoff to the follower, $c_e$, are the \emph{negative} of the edges costs that are denoted below each edge. That is, this is an instance where the follower is responding by choosing a shortest path with respect to the edge weights that correspond to the probability with which the leader plays them. 
\begin{figure}[h]
        \centering
\scalebox{0.93}{
\begin{tikzpicture}[scale=0.5]
    \node[anchor=south west,inner sep=0] (image) at (0,0) {     \includegraphics[width=0.5\textwidth]{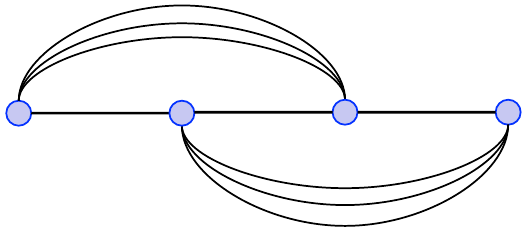}};
    \begin{scope}[x={(image.south east)},y={(image.north west)}]
	\node() at (0.035, 0.535)[]{$s$};
	\node() at (0.955, 0.54)[]{$t$};
	\node() at (0.34, 0.535)[]{$a$};
	\node() at (0.65, 0.54)[]{$b$};

        \node() at (0.2, 0.46)[]{$1$};
        \node() at (0.5, 0.46)[]{$1$};
        \node() at (0.8, 0.46)[]{$1$};
        \node() at (0.65, 0.03)[]{$2.4$};
        \node() at (0.34, 0.78)[]{$2.2$};
        \node() at (0.2, 0.6)[]{${\color{gray}0.4}$};
        \node() at (0.8, 0.6)[]{${\color{gray}0.6}$};
    \end{scope}
\end{tikzpicture}   
}  
\caption{An example where additional commitment increases the leader's payoff in the Stackelberg equilibrium. The follower's cost for each edge is denoted \emph{below} the edge, i.e., $-c_{sa}=-c_{ab}=-c_{bt}=1$ and $-c_{ab} = 2.2$ and $-c_{at} = 2.4$ for all the edges between $s$ and $b$, and between $a$ and $t$.
The mixed strategy of the leader is denoted in gray \emph{above} the edges.} 

\label{fig:example-commit}
\end{figure}
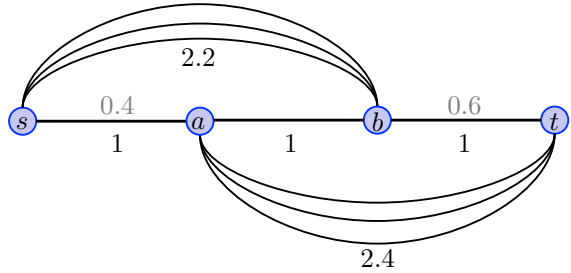        
Note that, since there are  many parallel edges $ab$ and $at$, the leaders optimal strategy (with or without additional commitment) is only supported on  edges, $sa$, $ab$, and $bt$. It is not hard to see that without any additional commitment, the Stackelberg equilibrium involves the leader playing edge $sa$ with probability $x_{sa} = 0.4$ and edge ${bt}$ with probability $x_{bt} = 0.6$, and all other edges with probability $x_e = 0$. Note that in such a mixed strategy the follower chooses path $a, b, t$ and the leader's payoff is $0.6$.
On the other hand, when the leader commits to providing additional incentive (or discount in the cost of a path) of $v_{sabt} = 0.2$ on path $s, a, b, t$, the follower best responds by choosing path $s, a, b, t$ and the leader's payoff is $0.8$.
\end{example}

Our main theorem in this section show that there is a polynomial time algorithm for finding the Stackelberg equilibrium of this modified game.


\begin{theorem} \label{thm:main-discount}
There is a polynomial time algorithm for finding a Stackelberg equilibrium if one can solve the following problem in polynomial time: Given $\vec x$ and value $W$, return  $S\in \S$, such that $\sum_{e\in S} (-x_e + c_e) \leq - W$, or return ``None'' if no such $S\in \S$ exists.\footnote{An example of a game where this linear program can be solved efficiently is the shortest path game in Example~\ref{eg:commit}.}
\end{theorem}


%
At a high level, we show that a Stackelberg equilibrium, $(\vec x^*, \vec V^*)$, can be found by finding the optimal solution $(\vec x, \vec 0)$ (with no additional incentives) that involves maximizing the followers payoff of the best response, and then providing enough incentive on one of the follower's actions. In particular, we choose to provide incentive on the specific $S\in \S$ that constitutes the best response of the  follower to the mixed strategy $(\vec 0, \vec 0)$.

For the first step of this proof, we consider the following LP, which can be efficiently solved by the separation oracle  given in Theorem~\ref{thm:main-discount},
\begin{equation} \label{eq:LP}
\begin{split}
 &\max_{\vec x, W} ~ W +  \sum_{e\in E} x_e C_e \\
 &\forall S\in \S,~\sum_{e\in S} (-x_e+ c_e)   \leq -W.
 \end{split}
\end{equation}
Let $\vec x^*,W^*$ be the solution to the above LP. Furthermore, let $S^* = \arg\max_{S\in \S} \sum_{e\in S} c_e$, and consider the incentive commitments $V^*_{S^*} = -  W^* - \sum_{e\in S^*} (-x_e+ c_e)$, and $V^*_{0} = 0$ for all $S\neq S^*$. That is, we provide enough incentive on set $S^*$ such that it becomes the best response for the follower.

To prove Theorem~\ref{thm:main-discount}, we first prove a lower bound on the incentive needed to make an action the best response of the follower.
\begin{lemma}\label{lem:v_s}
Let $(\vec x', \vec V')$ be any mixed strategy of the leader
and let $S' = b(\vec x', \vec V')$ be the corresponding best response of the follower.
Let $W' = - \max_{S\in \S} \sum_{e\in S} (-x'_e + c_e)$. We have,
\[  V'_{S'}  \geq -W' -  \sum_{e\in S'} (-x'_e + c_e). 
\]
\end{lemma}
\begin{proof}
Let $S'' = \arg\max_{S\in \S} \sum_{e\in S} (-x'_e + c_e)$ be one of the actions of the follower that corresponds to value $-W'$. Since all incentives are non-negative,
we have that 
\[ \payofff( (\vec x', \vec V'), S'') =  \sum_{e\in S} (-x_e + c_e) + V'_{S''} \geq -W'.
\]
Therefore, $S'$ is a best response to $(\vec x', \vec V')$ only if
\[
 \sum_{e\in S'} (-x'_e + c_e) + V'_{S'} = \payofff( (\vec x', \vec V'), S') \geq  \payofff((\vec x', \vec V'), S'')  \geq -W'.
\]
This completes the proof.
\end{proof}

\begin{proof}[proof of Theorem~\ref{thm:main-discount}]
Let $(\vec x^*, W^*)$ be the solution to Equation~\ref{eq:LP}. Let $S^* =  \arg\max_{S\in \S} \sum_{e\in S} c_e$, and  let $V^*_{S^*} = -W^* - \sum_{e\in S^*} (-x^*_e+ c_e)$, and $V^*_{S} = 0$ for all $S\neq S^*$. It is clear that $b(\vec x^*, \vec V^*) = S^*$. Here, we show that $(\vec x^*, \vec V^*)$ is indeed the optimal leader strategy.

For any leader strategy $(\vec x', \vec V')$, let $S' =b(\vec x', \vec V')$ be the follower's best response. Moreover, let $W' = - \max_{S\in \S} \sum_{e\in S} (-x'_e + c_e)$. We have
\begin{align}
\payoffl((\vec x^*, \vec V^*), S^*) &=\sum_{e\in S^*} x^*_e + \sum_{e\in E} x^*_e C^*_e - V^*_{S^*} \\
& = \sum_{e\in S^*} x^*_e + \sum_{e\in E} x^*_e C^*_e + W^* + \sum_{e\in S^*} (-x^*_e+c_e) \\
& = \sum_{e\in E} x^*_e C^*_e + W^* + \sum_{e\in S^*} c_e \\
& \geq \sum_{e\in E} x'_e C'_e + W' + \sum_{e\in S'} c_e,   \label{eq:proof1}
\end{align}
where the second equation is by the definition of $V^*_{S^*}$) and the last inequality follows by the fact that $(\vec x' , W')$ form a valid solution for the LP in Equation~\ref{eq:LP}, for which $(\vec x^*, W^*)$ is the optimal solution and the fact that $S^*$ is chosen to maximize $\sum_{e\in S^*} c_e$.

Using Lemma~\ref{lem:v_s} on the value of $V'_{S'}$, we have
\begin{align}
\payoffl((\vec x', \vec V'), S') & = \sum_{e\in S'} x'_e + \sum_{e\in E} x'_e C'_e  - V'_{S'} \\
 	&\leq \sum_{e\in S'} x'_e + \sum_{e\in E} x'_e C'_e  + W' + \sum_{e\in S'} (-x'_e + c_e)\\
    & =\sum_{e\in E} x'_e C'_e + W' + \sum_{e\in S'} c_e.    \label{eq:proof2}
\end{align}
Equations~\ref{eq:proof1} and \ref{eq:proof2}  complete the proof.

\end{proof}

%% file: sources/apxhardness_new.tex
\section{The \mgame\ Game}\label{section:apx}

In this section, we introduce a large but structured general-sum Stackelberg game, called \mgame, and examine the computational complexity of computing its Stackelberg equilibrium. We show two sets of results for this game. In Section~\ref{section:apx-hard}, we show that this problem is APX-hard. This implies that unlike zero-sum games, finding a Stackelberg equilibrium is computationally hard even if best-response oracles are provided. In Section~\ref{sec:mgame-approx}, we give an efficient $\frac{1}{12}$ approximation for this problem.

The \mgame\ game is defined as follows. 
Consider the leader and follower, $\Le$ and $\Fo$.
Consider a multigraph $G = (V, E)$ and a one-to-one mapping (permutation) $\pi: E\rightarrow E$.  
Note that $\pi$ may take different values on parallel edges of a multi-graph. In the remainder of this section, we refer to a multi-graph $G$ as a \emph{graph}.
In \mgame, the set of pure strategies of both players is the set of all matchings in $G$. Given matchings $M_\Le$ and $M_\Fo$ played by the leader and follower, respectively, we define
\begin{equation*}
\payoffa(M_\Le,M_\Fo) = |M_\Le \cap \pi(M_\Fo)|, \text{and} \quad
\payoffb(M_\Le,M_\Fo) = |M_\Le \cap M_\Fo|,
\end{equation*} 
where for a set $S \subseteq E$, we define $\pi(S) = \{e \in S| \pi(e)\}$.
Note that $G$ and $\pi$ are fixed and known to both players.

Let us highlight an important aspect of our hardness result in advance. As the next observation shows, the strategy space of the players in \mgame, though large, is very structured. At a high level, the reward of each player is a linear function of the action of the other and each player can efficiently optimize a linear function over the strategy space of the other player, for example, each player can compute a best-response  to a mixed strategy of the other. 
\begin{observation}\label{salam1}
There is a polynomial time algorithm such that for every vector $\vec w \in [0,1]^{|E|}$ finds a strategy of the players whose corresponding representation vector $\vec v$ maximizes $\vec v \cdot \vec w$.
\end{observation}
\begin{proof}[Sketch]
This problem reduces to computing a maximum weighted matching of a graph with edge weights  $w_e$  for all $e\in E$, which can be performed efficiently~\cite{cormen2009introduction}.
\end{proof}

In a zero-sum game, existence of such a structure leads to efficient algorithms for computing the Nash or Stackelberg equilibria~\cite{ahmadinejad2016duels, DHL+17, KV05}. On the other hand, our APX-hardness result for the \mgame\ game shows that existence of this structure does not necessarily lead to efficient algorithms for computing Stackelberg equilibria in general-sum games. With this in mind, we present our hardness results next.

\subsection{Hardness of approximation}\label{section:apx-hard}

In this section, we show that it is impossible to approximate a Stackelberg equilibrium of the \mgame\  game in polynomial time within an arbitrarily small constant factor unless P=NP.

Before we proceed to the proof, we define an auxiliary problem and show a hardness result for this problem. Then, we take advantage of this hardness result and show that computing a Stackelberg equilibrium of \mgame\enspace is APX-hard. We call the intermediary problem \matchingproblem\ and define it as follows:
\begin{center}
    \noindent\framebox{\begin{minipage}{3in}
            {\centering \matchingproblem \\[0.25cm]}
            \textsc{Input}: An unweighted undirected graph $G = (V, E)$,\\
            \quad             and a permutation $\pi: E \rightarrow E$. \\[0.25cm]
            \textsc{Output:} A matching $M$ of $G$ that maximizes $|M \cap \pi(M)|$.
    \end{minipage}}
\end{center}

For an instance $\I$ of \matchingproblem, we denote by $\opti$ the optimal solution to $\I$ and refer to the value of this solution by $\vali$.

We show that $\pi$\textsc{-transformation-identical-matching} has a \textit{hard gap} at \textit{gap location $1$}. That is, it is NP-hard to decide whether for a given graph $G$ with $n$ vertices and a function $\pi$, the solution of the \matchingproblem\enspace problem is exactly equal to $n/2$ or at most $(1-\epsilon)n/2$ for some $\epsilon > 0$.

\begin{lemma}\label{lm1}
    There exists an $\epsilon > 0$ such that it is NP-hard to decide whether the solution of the \matchingproblem\enspace problem is exactly equal to $n/2$ or less than $(1-\epsilon)n/2$ where $n$ is the number of the vertices of the input graph.
\end{lemma}
We defer the proof of \ref{lm1} to the end of this section and first show how this lemma can be used to prove the main result of this section.

\begin{theorem}\label{apxhardness} \label{thm:apxhardness}
    Computing a Stackelberg equilibrium of \textsc{Permuted Matching} is APX-hard.
\end{theorem}
\begin{proof}
    More generally, we show that approximating a Stackelberg equilibrium of the \mgame\  game has a hard gap at gap location $1$. This immediately implies a hardness of approximation. We show this by a reduction from the \matchingproblem\enspace problem. Suppose we are given an instance $\I = \langle G, \pi\rangle$ of the $\pi$\textsc{-transformation-identical-matching} problem and wish to decide for some $\epsilon' > 0$, whether the solution of this problem achieves a value that is exactly $n/2$ or is bounded above by $(1-\epsilon')n/2$ where $n$ is the size of $G$. 
    Based on $\I$, we construct an instance $\cori$ of the \mgame\  game with the same graph $G$ and permutation $\pi$ and seek to find a Stackelberg equilibrium in this game. Note that by definition, $\vali$ is equal to $n/2$ if and only if $G$ contains a perfect matching that is identical to its $\pi$-transformation. Otherwise, $\vali$ is at most $(1-\epsilon')n/2$ and thus any matching of $G$ shares no more than $(1-\epsilon')n/2$ edges with its $\pi$-transformation.
    
    Since for small enough $\epsilon'$, it is NP-hard to distinguish the two cases (Lemma \ref{lm1}), we show that it is NP-hard to approximate a Stackelberg equilibrium of the leader in $\cori$. If $\vali=n/2$, then there exists a perfect matching in $G$ that is identical to its $\pi$-transformation. Thus, if both players play this matching in $\cori$, they both get a payoff of $n/2$. Notice that $n/2$ is the maximum possible payoff for any player in this game, therefore, such a strategy pair is a Stackelberg equilibrium. Hence, in case $\vali = n/2$, the leader achieves a payoff of $n/2$ in a Stackelberg equilibrium of the corresponding \mgame\  game. 
    
%
    Now, suppose for $\epsilon < \epsilon'/13$ we have a $1-\epsilon$ approximation solution for $\cori$. If $\vali = n/2$, then the payoff of the leader in an exact solution of $\cori$ is $n/2$ and therefore a $1-\epsilon$ approximation solution guarantees a payoff of at least  $n(1-\epsilon)/2$ for the leader. Let the strategies of the leader and follower be $X$ and $y$ in such a solution. Therefore, $\payoffl(X,y) \geq n(1-\epsilon)/2$. Notice that $X$ may be a mixed strategy, but we can assume w.l.g that $y$ is a pure strategy since there always exists a best response for the follower which is pure. Also, let $y^*$ be the $\pi$-transformation of strategy $y$. Let for two matchings $x$ and $y$, $\common(x,y)$ denote the number of edges that $x$ and $y$ have in common and define $\dist(x,y) = |x| + |y|-2\common(x,y)$. Recall that the payoff of the leader in this game can be formulated as $\mathbb{E}_{x \sim X} [\common(x,y^*)]$. Since this value is at least $n(1-\epsilon)/2$ we have:
    \begin{equation*}
    \mathbb{E}_{x \sim X} [\common(x,y^*)] = \payoffl(X,y) \geq n(1-\epsilon)/2
    \end{equation*}
    and thus
    \begin{equation}\label{fact1}
    \begin{split}
    \mathbb{E}_{x \sim X} [\dist(x,y^*)] & =  \mathbb{E}_{x \sim X} [|x| + |y^*| - 2\common(x,y^*)]\\
    & \leq  \mathbb{E}_{x \sim X} [n - 2\common(x,y^*)]\\
    & = n - 2\mathbb{E}_{x \sim X} [\common(x,y^*)]\\
    & \leq n - 2 n(1-\epsilon)/2\\
    & = n \epsilon.
    \end{split}
    \end{equation}
    Inequality \eqref{fact1} shows that $y^*$ is very similar (in expectation) to a random matching drawn from strategy $X$. This intuitively implies that pure strategies of $X$ should have a considerable amount of edges in common. It follows from the definition that for three matchings $x$, $y$, and $z$ we have $\dist(x,z) \leq \dist(x,y) + \dist(y,z)$. Therefore, we have 
    \begin{equation}\label{fact2}
    \begin{split}
    \mathbb{E}_{x \sim X, x' \sim X} [\dist(x,x')] & \leq \mathbb{E}_{x \sim X, x' \sim X} [\dist(x,y^*) + \dist(y^*,x')]\\
    & = \mathbb{E}_{x \sim X, x' \sim X} [\dist(x,y^*) + \dist(x', y^*)]\\
    & = \mathbb{E}_{x \sim X} [\dist(x,y^*)] + \mathbb{E}_{x' \sim X} [\dist(x',y^*)]\\
    & = 2 \mathbb{E}_{x \sim X} [\dist(x,y^*)]\\
    & \leq 2\epsilon n
    \end{split}
    \end{equation}
    Recall that the payoff of the follower is determined by the number of edges his matching shares with that of the leader. Moreover, since $\payoffl(X,y) \geq n(1-\epsilon)/2$, this implies that $\mathbb{E}_{x \sim X}[|x|] \geq n(1-\epsilon)/2$. What Inequality \eqref{fact2} implies is that if the follower plays $X$ instead of $y$, he gets a payoff of at least $\mathbb{E}_{x \sim X}|x| - 2\epsilon n \geq (1-5\epsilon)n/2$ against $X$. In other words $\payofff(X,X) \geq (1-5\epsilon)n/2$. Since $y$ is a best response of the follower against the leader's strategy, we have $\payofff(X,y) \geq \payofff(X,X) \geq (1-5\epsilon)n/2$ and thus
    \begin{align*}
        \mathbb{E}_{x \sim X} [\common(x,y)]  &= \payofff(X,y) \\
          & \geq \payofff(X,X)  \\
          &=  \mathbb{E}_{x \sim X, x' \sim X} [\common(x,x')] \\
          & \geq (1-5\epsilon)n/2.
    \end{align*}
    Hence
    \begin{equation}\label{fact3}
    \begin{split}
    \mathbb{E}_{x \sim X} [\dist(x,y)] & = \mathbb{E}_{x \sim X} [|x| + |y| - 2\common(x,y)] \\
    & \leq \mathbb{E}_{x \sim X} [n - 2\common(x,y)]\\
    & = n-2\mathbb{E}_{x \sim X} [\common(x,y)]\\
    & \leq n-2(1-5\epsilon)n/2\\
    & \leq 5\epsilon n.
    \end{split}
    \end{equation}
    Combining Inequalities \eqref{fact1} and \eqref{fact3} yields
    \begin{equation*}
    \dist(y, y^*) \leq \mathbb{E}_{x \sim X} [\dist(x,y)] + \mathbb{E}_{x \sim X} [\dist(x,y^*)] \leq 6\epsilon n.
    \end{equation*}
    Therefore, we have $\common(y, y^*) \geq |y^*|- 6\epsilon n$ and since $|y^*| \geq (1-\epsilon)n/2$ we have $\common(y, y^*) \geq (1-13\epsilon)n/2  >  (1-\epsilon')n/2$. If $\vali \neq n/2$, then $\vali$ is bounded by $(1-\epsilon')n/2$. Therefore, $\common(y, y^*) > (1-\epsilon')n/2$ holds if and only if $\vali=n/2$. Thus, an approximation solution for $\cori$ within a factor $(1-\epsilon) > (1-\epsilon'/13)$ can be used to decide if the solution of $\I$ is $n/2$ or bounded by $(1-\epsilon')n/2$. This implies a hard gap for the \matchingproblem\enspace problem at gap location $1$.
\end{proof}

All that remains is to prove the statement of Lemma \ref{lm1} and that completes the proof Theorem \ref{apxhardness}.

\begin{proof}[proof of Lemma \ref{lm1}]
    We show this lemma by a reduction from the \threedM\enspace problem. In the \threedM\enspace problem, we are given a hypergraph $G$ whose vertices are divided into three parts $A$, $B$, and $C$. Every hyper-edge of $G$ is a triple $(a,b,c)$ of the vertices such that $a \in A$, $b \in B$, and $c \in C$ hold. A matching in this graph is a subset of the hyper-edges that do not share any vertices. The goal of the problem is to find a matching with the maximum number of hyper-edges. \petrankhardness~\cite{petrank1994hardness} showed that the \threedM\enspace problem has a hard gap at gap location $1$.

\input{figs/reduction2.tex}

    We show via a reduction that \matchingproblem\enspace is harder than the \threedM\enspace problem. This implies a similar hardness result for the \matchingproblem\enspace problem. To this end, suppose we are given an instance $\I$ of the \threedM\enspace problem. Let the hyper-graph of this instance be $G$ and its vertices lie in three parts $A$, $B$, and $C$ such that every hyper-edge of the graph contains a vertex of each part. We construct a bipartite graph $G'$ and a function $\pi$ based on $G$ as follows:
    \begin{itemize}
     \setlength\itemsep{2pt}
        \item $G'$ contains two independent parts $X$ and $Y$, each of which contains an endpoint of every edge of the graph.
        \item $X = A' \cup A''$ where  $A'$ and $A''$ are two copies of $A$. That is, for every vertex $a \in A$, we put two vertices $a'$ and $a''$ in $X$ (in $A'$ and $A''$ respectively).
        \item $Y = B' \cup C'$ where $B'$ is a copy of $B$ and $C'$ is a copy of $C$. That is, for every vertex $b \in B$ and $c\in C$ we put a vertex $b'$ and $c'$ in $Y$ (in part $B'$ and $C'$, respectively).
       \item For every edge $(a,b,c)$ of $G$, we put two edges $(a',b')$ and $(a'',c')$ in $G'$. Moreover, we set $\pi((a',b')) = (a'',c')$ and $\pi((a'',c')) = (a',b')$. Note that, there may be multiple edges between two vertices of multi-graph $G$ with different $\pi$ values.
    \end{itemize}
    Now, we argue that for every 3-dimensional matching $M$ of $G$, there exists a matching $M'$ of $G'$ such that $|M' \cap \pi(M')| = 2|M|$ and viceversa. To this end, suppose $M$ is a 3-dimensional matching of $G$. Now, we set $M' = \{(a,b,c) \in E(G) | (a',b')\} \cup \{(a,b,c) \in E(G) | (a'',c')\}$. Since $M$ is a matching, no two edges of $M'$ share a vertex and thus $M'$ is also a matching. Moreover, for every edge in $M'$, its $\pi$ transformation is also included in $M'$ and therefore $M' = \pi(M')$. Hence $|M' \cap \pi(M')| = |M'| = 2|M|$. A similar argument shows that for any matching $M'$ of $G'$ such that $|M' \cap \pi(M')| = 2k$, there exists a matching of size $k$ in $G$. Therefore, the problem of finding a maximum 3-dimensional matching of $G$ reduces to finding a matching of $G'$ that shares the maximum number of edges with its $\pi$-transformation. Since the \threedM\enspace problem has a hard gap at gap location $1$, so does the \matchingproblem\enspace problem. Figure \ref{reduction1} describes the reduction mentioned above.
\end{proof}

%% file: figs/reduction2.tex
\begin{figure*}[h]
        \centering
\begin{tikzpicture}
    \node[anchor=south west,inner sep=0] (image) at (0,0) {     \includegraphics[width=0.9\textwidth]{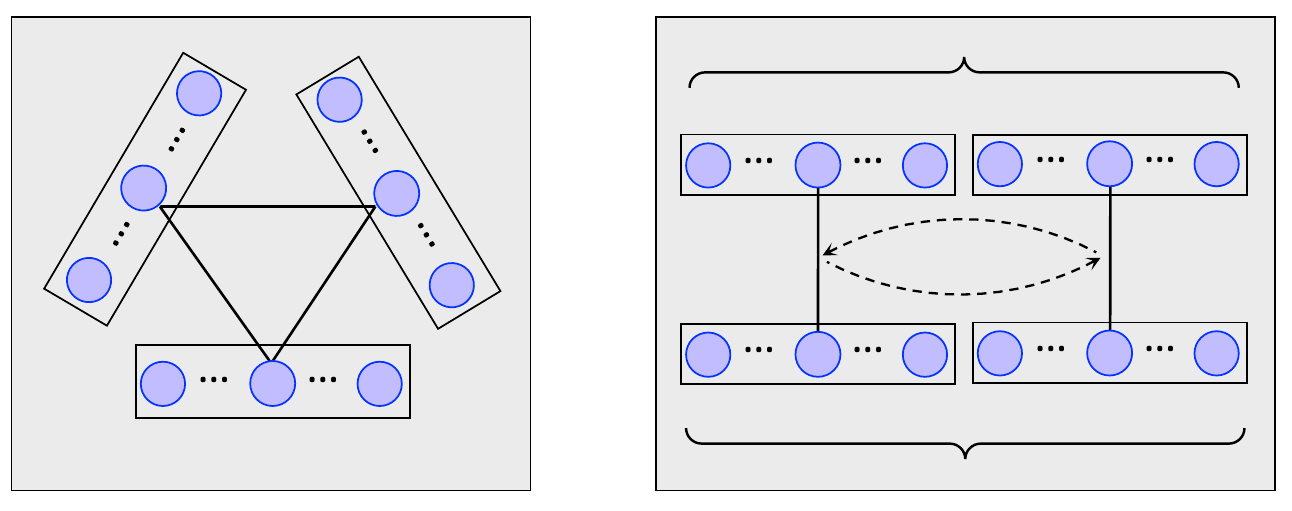}};
    \begin{scope}[x={(image.south east)},y={(image.north west)}]
	\node() at (0.113, 0.63)[]{$b$};
	\node() at (0.075, 0.7)[]{$B$};
	\node() at (0.306, 0.62)[]{$c$};
	\node() at (0.345, 0.7)[]{$C$};
	\node() at (0.212, 0.25)[]{$a$};
	\node() at (0.213, 0.12)[]{$A$};
	\node() at (0.213, 0.12)[]{$A$};

\draw [-to,thick,snake=snake,segment amplitude=.4mm,
segment length=2mm,line after snake=1mm,color=red](0.4,0.5) -- (0.52, 0.5);
	\node() at (0.395, 0.5)[]{$\I$}; 
	\node() at (0.555, 0.5)[]{$\cori$}; 
	\node() at (0.745, 0.5)[]{\small $\pi$-transformation};
	
	\node() at (0.745, 0.925)[]{$X$};
	\node() at (0.75, 0.07)[]{$Y$};
	\node() at (0.64, 0.18)[]{$B'$};
	\node() at (0.865, 0.18)[]{$C'$};
	
	\node() at (0.64, 0.81)[]{$A'$};
	\node() at (0.865, 0.81)[]{$A''$};
	
	\node() at (0.635, 0.3)[]{$b'$};
	\node() at (0.635, 0.68)[]{$a'$};
		
	\node() at (0.86, 0.31)[]{$c'$};
	\node() at (0.86, 0.68)[]{$a''$};
    \end{scope}

\end{tikzpicture}\caption{The figure on the left shows an instance $\I$ of the  \threedM\enspace problem and the figure on the right shows the corresponding instance $\cori$ of the \matchingproblem\enspace problem. Solid segments show the edges of the graphs and dashed segments specify the corresponding $\pi$-transformations of the edges.}\label{reduction1}
\end{figure*}
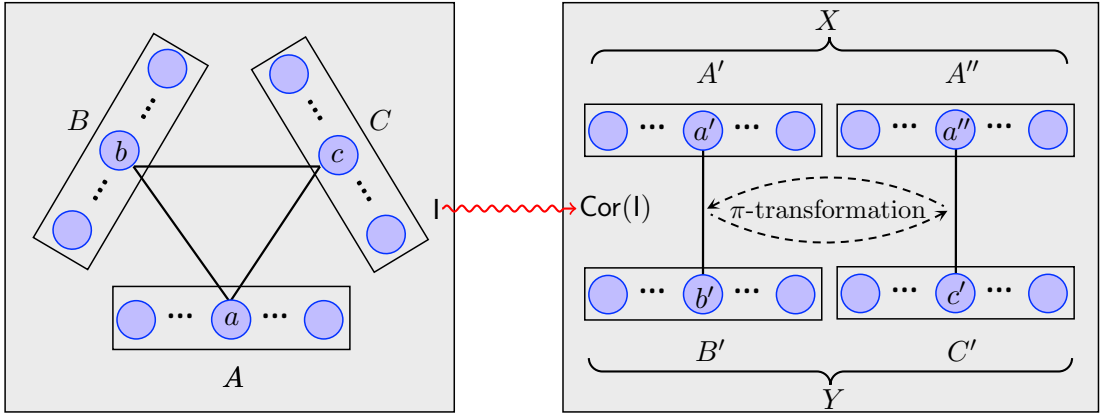

%% file: sources/app_example.tex
\section{Stackelberg, Nash, Minimiax Equilibria in General-Sum Games}
\label{app:example}
In general-sum games, Stackelberg, Nash, and Minimiax equilibria diverge.
Let us demonstrate this by the game matrix in Figure~\ref{tab:general-sum}. In the Stackelberg equilibrium, the row  player (leader) commits to mixed strategy $(\frac 12, \frac 12)$ and the column player (follower) responds by $(0, 1)$, i.e., playing action $R$ deterministically. In this case, the leader gets utility $7.5$.
In the Nash equilibrium, however, both players move simultaneously and  play strategy $(1, 0)$. In this case, the row player gets utility $1$. 
In the Minimax setting, each player maximizes its own utility assuming the worst-case.\footnote{This is equivalent to each player playing their side of a Nash equilibrium when they perceive the game as a zero-sum game.} In this case, the row player plays $(1, 0)$ and the column player plays $(\frac 12, \frac 12)$, with the row player gaining utility $5.5$.

\begin{table}\
\centering
	\begin{tabular}{r|c|c|}
			& L & R\\
		\hline U & 1, 1 &  10, 0\\ 
		\hline D & 0, 0 &  5,   1\\
		\hline
	\end{tabular}
	\qquad
		\begin{tabular}{r c c }
	    & Strategy & Row player's payoff  \\ 	\hline	
   Stackelberg: & $(\frac 12, \frac 12) \times (0,1)$ & 7.5 \\
	   Nash:    &   $(1, 0) \times (1, 0)$&   1  \\
	   Minimax: & $(1,0) \times (\frac 12, \frac 12)$ & 5.5 \\
	\end{tabular}
   	\caption{{\small A general-sum game matrix where Nash, Minimax, and Stackelberg equilibria are different.}}\label{tab:general-sum}
\end{table}

%% file: sources/apx-appendix.tex
\section{A $1/12$ Approximate Solution for The \mgame\ Game}\label{sec:mgame-approx}
In Section \ref{section:apx-hard}, we showed that approximating a Stackelberg Equilibrium of the \mgame\  game within an arbitrarily small error is NP-hard. We complement this result by presenting a polynomial time algorithm for computing a $1/12$-approximately optimal solution. In other words, our algorithm provides a strategy $X$ for the leader that achieves at least a $1/12$ fraction of the payoff of a Stackelberg equilibrium  against the best response of his opponent.

We first devise a greedy algorithm for finding two matchings $x$ and $x'$ such that $x$ shares as many edges as possible with the $\pi$-transformation of $x'$, i.e.,  maximizing $|x\cap \pi(x')|$. To this end, we begin with two empty matchings $x$ and $x'$ and iteratively choose two edges $e$ and $e'$ of the graph such that $\pi(e') = e$, $e$ is not adjacent\footnote{two edges are adjacent if they share a common vertex.} to any edge of $x$, and $e'$ is not adjacent to any edge of $x'$. We then update the matchings by adding $e$ to $x$ and $e'$ to $x'$.

\begin{algorithm}

    \KwData{A graph $G$ and a permutation $\pi$ over the edges of $G$.}
    \KwResult{Two matchings $x$ and $x'$.}
    $x \leftarrow \emptyset$,   $x' \leftarrow \emptyset$\;
    \lWhile{\text{there exist edges $e$ and $e'$ s.t.} \begin{itemize}
     \setlength\itemsep{0pt}
                \item $e = \pi(e')$
                \item $e$ does not share any vertices with $x$ and $e'$ does not share any vertices with $x'$
            \end{itemize}}{
            \text{Add $e$ to $x$ and $e'$ to $x'$}
    }
    \textbf{Return} $(x,x')$\;
    \caption{Maximizing the number of edges shared between $x$ and $\pi(x')$}    \label{alg:matchings}
\end{algorithm}

We show that  $\payoffl(x, x')$ for $x, x'$ computed by  Algorithm~\ref{alg:matchings} is at least $1/4$ of the leader's payoff in a Stackelberg equilibrium.

\begin{observation}\label{observation:matchings}
    Let $x$ and $x'$ be the matchings determined by Algorithm \ref{alg:matchings}. $|x \cap \pi(x')|$ is at least $1/4$ fraction of the optimal payoff the leader gets in any Stackelberg Equilibria.
\end{observation}
\begin{proof}
    Recall that the payoff of the leader is the number of edges his matching has in common with the $\pi$-transformation of the follower's matching. This value is always upper bounded by $\max_{y,y'} |y \cap \pi(y')|$. 
    Assume to the contrary that a pair of matchings $y, y'$ exists such that $|y \cap \pi(y')| >  4|x \cap\pi(x')|$.
Note that by the choices of Algorithm~\ref{alg:matchings}, we have $|x| = |x'| = |x \cap\pi(x')|$. Moreover, there are at least $4|x \cap \pi(x')|$ edges in $y'$ whose $\pi$-transformations appear in $y$. Notice that for every edge $e'$ in $x'$, there are at most four such edges in $y'$ that either share an end-point with $e'$ or their $\pi$-transfrmations share an endpoint with $\pi(e')$. Thus, there exists an edge $e'$ in $y$ such that neither $e'$ shares a vertex with the vertices of $x'$ nor $\pi(e')$ shares an endpoint with the vertices of $x$. This is a contradiction, as these edges could have been added to $x$ and $x'$ by Algorithm~\ref{alg:matchings}. Therefore, $\max_{y,y'} |y \cap \pi(y')| \leq 4|x \cap\pi(x')|$. This completes the proof.
\end{proof}
Let $\epsilon > 0$ be an arbitrarily small positive number and $X$ be a strategy of the leader that plays matching $x$ with probability $1/3-\epsilon$ and matching $x'$ with probability $2/3+\epsilon$. Moreover, let $y$ be a best response of the follower against $X$. Based on Observation \ref{observation:matchings}, we show that $\payoffl(X,y)$ is at least $(1-3\epsilon)/12$ fractional of the optimal payoff the leader can get in any SE.

\begin{theorem}
    Let $X$ and $y$ be the strategies of the leader and the follower defined above. $\payoffl(X,y)$ is at least $(1-3\epsilon)/12$ fraction of the optimal payoff the leader achieves in any Stackelberg equilibrium.
\end{theorem}
\begin{proof}
Without loss of generality, $y$ is a pure strategy. We first show that $y$ contains all edges of $x'$. Assume on the  contrary that an edge $e'$ of $x'$ does not appear in $y$. If we add this edge to $y$ and remove the edges of $y$ that share an endpoint with $e'$ then the payoff of the follower changes as follows: Because $X$ plays $x'$ with probability $2/3+\epsilon$ then  $e'$ appears in his strategy with probability $2/3+\epsilon$ and thus the payoff of the follower is increased by $2/3+\epsilon$. Moreover, none of the crossing edges of $e'$ with edges of $y$ are in $x'$ and thus appear in the leader's strategy with probability at most $1/3-\epsilon$. This incurs a loss of at most $2(1/3-\epsilon) < 2/3+\epsilon$. Thus, this change improves the payoff of the follower which shows $y$ is not a best response against $X$ which is a contradiction.
    
    Note that $X$ plays $x$ with probability $1/3-\epsilon$ and $y$ contains every edge of $x'$. Let $U$ be the maximum payoff of the leader in any Stackelberg equilibrium. We have
    \begin{align*}
    \payoffl(X,y) &\geq (1/3-\epsilon)\payoffl(x,x') = (1/3-\epsilon) | x \cap \pi(x')| \\
    &= \frac 13  (1-3\epsilon)|x \cap \pi(x')| \geq \frac {1}{12} (1-3\epsilon) U,
    \end{align*}
 where the last step holds by  Observation \ref{observation:matchings}.
\end{proof}